\documentclass[11pt,a4paper]{article}%
\usepackage{amsfonts}
\usepackage{amsmath}
\usepackage{amssymb}
\usepackage{graphicx}
\usepackage{geometry}
\usepackage{booktabs} 
\usepackage{array}
\usepackage{cancel}%
\providecommand{\U}[1]{\protect\rule{.1in}{.1in}}

\newtheorem{corollary}{Corollary}

\newtheorem{proposition}{Proposition}

\newenvironment{proof}[1][Proof]{\textbf{#1.} }{\  \rule{0.5em}{0.5em}}
\makeatletter
\def \@removefromreset#1#2{\let \@tempb \@elt
\def \@tempa#1{@&#1}\expandafter \let \csname @*#1*\endcsname \@tempa
\def \@elt##1{\expandafter \ifx \csname @*##1*\endcsname \@tempa \else
\noexpand \@elt{##1}\fi}     \expandafter \edef \csname cl@#2\endcsname{\csname cl@#2\endcsname}     \let \@elt \@tempb
\expandafter \let \csname @*#1*\endcsname \@undefined}

\@removefromreset{equation}{section}

\@removefromreset{theorem}{section}
\makeatother
\usepackage{listings,xcolor}
\begin{document}

\title{Nonviolation of the CHSH inequality under local spin-$1$  measurements on two spin  qutrits}
\author{Louis Hanotel$^{1}$ and Elena R. Loubenets$^{1,2,*}$\\
$^{1}$Department of Applied Mathematics, MIEM, HSE University, \\
Moscow 123458, Russia\\
$^{2}$Steklov Mathematical Institute of Russian Academy of Sciences, \\
Moscow 119991, Russia\\\
*Author to whom any correspondence should be addressed $\ \ \ \ $\\
Email: elena.loubenets@hse.ru  $\ \ \ \ \ \ \ \ \ \ \ \ \ \ \ \ \ \ \ \ \ \ \ \ \ \ \ \ \ \ \ \ \ \ \ \ \ \ \ \ $}
\date{}
\maketitle

\begin{abstract}
In the present paper, based on the general analytical expression [arXiv:2412.03470] for the maximum of the CHSH expectation under local Alice and Bob spin-$s$ measurements in a two-qudit state of dimension $d=2s+1$, $s\geq 1/2$, we analyze whether or not, under spin-$1$ measurements in an arbitrary two-qutrit state, the CHSH inequality is violated. We find analytically for a variety of pure nonseparable two-qutrit states and also, numerically for $1,000,000$ randomly generated pure nonseparable two-qutrit states, that, under local Alice and Bob spin-$1$ measurements in each of these nonseparable states, including maximally entangled, the CHSH inequality is not violated. These results  together with the spectral decomposition of a mixed state lead us to the Conjecture that, under local Alice and Bob spin-$1$ measurements, every nonseparable two-qutrit state, pure or mixed, does not violate the CHSH inequality. For a variety of pure two-qutrit states, we further find the values of their concurrence and compare them with the values of their spin-$1$ CHSH parameter, which determines violation or nonviolation by a two-qutrit state of the CHSH inequality under spin-$1$ measurements. This comparison indicates that, in contrast to 
 spin-$\frac{1}{2}$ measurements, where the spin-$\frac{1}{2}$ CHSH parameter of a pure two-qubit state is increasing monotonically with a growth of its entanglement, for a pure two-qutrit state, this is not the case. In particular, for the two-qutrit GHZ state, which is maximally entangled, the spin-$1$ CHSH parameter is equal to $\sqrt{\frac{8}{9}}$, while, for some separable pure two-qutrit states, this parameter can be equal to unity. Moreover, for the two-qutrit Horodecki state, the spin-$1$ CHSH parameter is equal to $4\sqrt{2}/21<1$ regardless of the entanglement type of this mixed state.

\end{abstract}

\section{Introduction}
\label{introduction}
Among a variety of Bell inequalities\footnote{For Bell inequalities, either on correlation functions
or on joint probabilities, see \cite{Lou:08} and references therein.} the Clauser--Horn--Shimony--Holt (CHSH)
inequality
\cite{Cla.Hor.Shi:69} is one of the most applied in different quantum information
processing tasks. The violation of this inequality in the quantum case has been
analyzed in many articles (see \cite{Loub2020} and references
therein) and the following main results are known up to the moment. 

\begin{itemize}
\item A two-qudit state $\rho_{d_{1}\times d_{2}}$, $d_{1},d_{2}\geq2$, on $\mathbb{C}^{d_{1}}\otimes\mathbb{C}^{d_{2}}$ violates the CHSH
inequality iff the maximum $\Upsilon_{chsh}(\rho_{d_{1}\times d_{2}})$
of the absolute value of the quantum expectation:
\begin{align}
\left\langle \mathcal{B}_{chsh}(A_{1},A_{2};B_{1},B_{2})\right\rangle
_{\rho_{d_{1}\times d_{2}}}  & :=\mathrm{tr}[\rho_{d_{1}\times d_{2}%
}\mathcal{B}_{chsh}(A_{1},A_{2};B_{1},B_{2})],\label{chsh_expectation}\\
\mathcal{B}_{chsh}(A_{1},A_{2};B_{1},B_{2})  & =A_{1}\otimes B_{1}%
+A_{1}\otimes B_{2}+A_{2}\otimes B_{1}-A_{2}\otimes B_{2},\label{0_1}%
\end{align}
over \emph{all} Alice and Bob qudit observables $A_{i},B_{j},$ $i,j=1,2,$ with eigenvalues in $[-1,1],$
satisfies the
condition\footnote{In a local hidden variable (LHV) frame,
$\left\vert \left\langle \mathcal{B}_{chsh}(A_{1},A_{2};B_{1},B_{2}%
)\right\rangle _{\rho_{d_{1}\times d_{2}}}\right\vert \leq2$ and specifically
this inequality is referred to in quantum information as the CHSH
inequality. Under the
original derivation of this inequality \cite{Cla.Hor.Shi:69} within an
LHV frame quantum observables have eigenvalues
$\pm1.$}
\begin{equation}
\Upsilon_{chsh} (\rho_{d_{1}\times d_{2}})>2.\label{_0_2}
\end{equation}
For short, we further refer to  (\ref{chsh_expectation}) as the CHSH
expectation in a state $\rho_{d_{1}\times d_{2}}.$
For an arbitrary two-qudit state, the Tsirelson upper bound  \cite{7,8} reads $\Upsilon_{chsh}(\rho_{d_{1}\times d_{2}})\leq2\sqrt2$
and, besides the two-qubit Bell states, is attained \cite{Loub2020} at the maximally entangled pure two-qudit states $\rho_{d\times d}$ of an even dimension $d\geq4$, in particular, at \cite{Loub2020} the two-qudit Greenberger--Horne--Zeilinger state and at  \cite{singlet} the two-qudit singlet state.

\item For a pure two-qudit state $|\psi_{d\times d}\rangle
\langle\psi_{d\times d}|$, the maximum $\Upsilon^{(traceless)}_{chsh}%
(|\psi_{d\times d}\rangle)$ of the absolute value of the CHSH expectation (\ref{chsh_expectation}) over all Alice and Bob traceless qudit observables with eigenvalues in $[-1,1]$ admits the lower bound (see Eqs. (39) and (47) in \cite{Lou.Kuz.Han:24}):
\begin{equation}
\Upsilon^{(traceless)}_{chsh}(|\psi_{d\times d}\rangle)\geq2\sqrt{1+\frac
{1}{(2d-3)^{2}}\mathrm{C}^{2}(|\psi_{d\times d}\rangle)}, \label{0_6}
\end{equation}
where $\mathrm{C}(|\psi_{d\times d} \rangle)$ is the concurrence
\cite{Hil.Woo:97,Run.Buz.etal:01,Che.Alb.Fei:05} of a pure state $|\psi_{d\times d}\rangle$ and the equality holds\footnote{The  equality in a two-qubit case was first proved in
\cite{Ver.Wol:02} for the pure two-qubit state of specific form and
further in \cite{Lou.Kuz.Han:24} for every pure two-qubit state, see remark 1 in \cite{Lou.Kuz.Han:24}.} \cite{Lou.Kuz.Han:24, Ver.Wol:02} for any pure two-qubit state. 
Relation (\ref{0_6}) explicitly indicates that every nonseparable two-qudit state violates the CHSH inequality. This issue was first shown in \cite{1} via the choice for a given nonseparable state  of the specific qudit observables with eigenvalues $\pm1$ for which the absolute value of the CHSH expectation (\ref{chsh_expectation}) is more than two.  Note, however, that observables chosen in \cite{1} and further used in \cite{Lou.Kuz.Han:24} do not constitute  spin-$s$ observables, for the latter type of qudit observables see Section 2. 

\item 
For a state $\rho_{d\times d}$ of two spin qudits of dimension $d=2s+1,\ s\geq \ \frac{1}{2} $, the maximum $\Upsilon_{chsh}^{(spin-s)}%
(\rho_{d\times d})$ of the CHSH expectation over all local Alice and Bob spin-$s$ observables is given
by the following general expression (Theorem  1 in \cite{arXiv}):
\begin{equation}
\Upsilon_{chsh}^{(spin-s)}(\rho_{d\times d})=2\sqrt{z_{s}^{2}(\rho_{d\times
d})+\widetilde{z}_{s}^{2}(\rho_{d\times d})}\ ,\label{0_7}%
\end{equation}
where $z_{s}(\rho_{d\times d})$ and $\widetilde{z}_{s}(\rho_{d\times d})$ are two largest singular values of \emph{the spin-}$s$\emph{
correlation matrix} $\mathcal{Z}_{s}(\rho_{d\times d})$ of a state
$\rho_{d\times d}$, introduced in  \cite{arXiv} and defined via the relation
\begin{equation}
\mathcal{Z}_{s}^{(ij)}(\rho_{d\times d}):=\mathrm{tr}[\rho_{d\times d}%
\{S_{i}\otimes S_{j}\}]\in\mathbb{R},\text{ \ \ }i,j=1,2,3,\label{_0_8}%
\end{equation}
where $S_{i}$, $i=1,2,3,$ are the components of the qudit spin $S=(S_{1},S_{2},S_{3})$ with the eigenvalues $\left\lbrace-s,-(s-1),...,-1,\right.$ $\left.0,1,...,(s-1),s\right\rbrace$ including zero if $d$ is
odd, and $\{-s,-(s-1),...,-\frac{1}{2},\frac{1}{2},...,(s-1),s\}$ if $d$ is
even.
A two-qudit state $\rho_{d\times d}$ violates the CHSH inequality under
local Alice and Bob spin-$s$ measurements iff \emph{ its spin-$s$
CHSH parameter} $\gamma_{s}(\rho_{d\times d})$ satisfies the relation (Corollary 1 in \cite{arXiv}):
\begin{equation}
\gamma_{s}(\rho_{d\times d})=\frac{1}{s^{2}}\sqrt{z_{s}^{2}(\rho_{d\times
d})+\widetilde{z}_{s}^{2}(\rho_{d\times d})}>1\ .\label{0__9}%
\end{equation}

The analytical expression (\ref{0_7}) includes as a particular case the
expression in \cite{Hor.Hor.Hor:95} for the maximum of the CHSH expectation (\ref{chsh_expectation}) under spin-$\frac{1}{2}$ measurements, derived by Horodecki et al. in 1995. The spin-$\frac{1}{2}$
CHSH parameter $\gamma_{s=\frac{1}{2}}(\rho_{2\times2})$ coincides with parameter $M(\rho_{2\times2})$ 
given by Eq. (5) in \cite{Hor.Hor.Hor:95}. Note that, for a pure two-qubit state $|\psi_{2\times2}\rangle$,
\begin{equation}
\gamma_{s=\frac{1}{2}}(|\psi_{2\times2}\rangle
)=M(|\psi_{2\times2}\rangle)=\sqrt{1+\mathrm{C}^{2}(|\psi_{2\times2}\rangle)} \ ,\label{0_5}%
\end{equation}
where $\mathrm{C}(|\psi_{2\times2}\rangle)$ is the concurrence
\cite{Hil.Woo:97,Run.Buz.etal:01,Che.Alb.Fei:05} of this pure state.
\end{itemize}

Recall that, up to a real coefficient, every traceless qubit observable has the
form $\sigma_{n}:=n\cdot\sigma=\sum_{k}n_{k}\sigma_{k}$, where $n$ is a unit
vector in $\mathbb{R}^{3}$ and $\sigma:=(\sigma_{1},\sigma_{2},\sigma_{3})$ and every spin-$1/2$ observable of a spin qubit is given by $\frac{1}%
{2}\sigma_{n}$. Therefore, for a state $\rho
_{2\times2}$ of two spin-$\frac{1}{2}$ qubits, relation (\ref{0__9}) constitutes also the necessary
and sufficient condition for violation of the CHSH inequality under Alice and Bob measurements on traceless qubit observables with eigenvalues in $[-1,1]$.

This is not the case for $d=2s+1>2$  where spin observables are included into the set of all traceless qudit observables only as a particular subset, so that, for a state $\rho_{d\times d}$ of two spin-$s$ qudits with $s\geq1$:
\begin{equation}
\frac{1}{s^{2}}\Upsilon_{chsh}^{(spin-s)}(\rho_{d\times d})<\Upsilon^{(traceless)}_{chsh}(\rho_{d\times d}),
\end{equation}  
where the maximum $\Upsilon_{chsh}^{(spin-s)}(\rho_{d\times d})$ is given by the general analytical expression (\ref{0_7}) whereas finding an explicit general analytical expression for $\Upsilon^{(traceless)}_{chsh}(\rho_{d\times d})$, $s\geq1$, is an open problem\footnote{For the expression of this maximum via the generalized Gell-Mann representation of traceless qudit observables, see \cite{Loub2020}.}. 

Note that spin constitutes an important intrinsic feature of a qudit with dimension $d=2s+1$, $s\geq 1/2$, and the analysis of spin-$s$ measurements in the context of Bell inequalities is now relevant for a variety of problems, including quantum key distribution \cite{qkd}, phase transitions in fermionic systems \cite{ferm}, squeezing of spin states \cite{sq}, and the study of quantum correlations in high-energy physics, see in \cite{LHC, rev} and references therein.  The study of spin-$1$ systems have proven to be valuable for understanding quantum correlations in different decay processes \cite{higgs, BellLHC, vbs, charm, boj, bff}. Spin-$1$ systems have also garnered significant attention over the past decade due to their potential to develop quantum computation beyond qubit-based systems. 
 
However, up to the moment the problem whether or not the CHSH inequality is violated under Alice and Bob high spin $s\geq1$ measurements is still open -- though specifically this Bell inequality has been used for proving \cite{1, Lou.Kuz.Han:24}  nonlocality of every pure nonseparable two-qudit state.

In the present paper, based on the explicit general analytical expressions (\ref{0_7}) and (\ref{0__9}), derived in \cite{arXiv}, we analyze the solution of this problem for spin-$1$ measurements.

We find analytically for a variety of nonseparable two-qutrit states,
pure and mixed, and also, numerically for 1,000,000 randomly generated
nonseparable pure two-qutrit states that \emph{their spin-$1$ CHSH parameter
}$\gamma_{s=1}(\rho_{3\times3})\leq1.$ 
Based on this, we put forward the Conjecture that, under
local Alice and Bob spin-$1$ measurements in an arbitrary nonseparable two-qutrit
state, pure or mixed, the CHSH inequality is not violated. 

For a variety of pure two-qutrit states, we also further find  analytically the values of their  concurrence and compare them with the values of the spin-$1$ CHSH parameter of these states. This comparison indicates that, in contrast to the
situation for a pure two-qubit state, described in Eq.(\ref{0_5}),  in case of a pure two-qutrit state $|\psi_{3\times3}%
\rangle,$ \emph{the spin-$1$ CHSH parameter }$\gamma_{s=1}(|\psi_{3\times
3}\rangle)$ and hence, the maximum $\Upsilon_{chsh}^{(spin-1)}(|\psi_{3\times3}\rangle)$ of the CHSH expectation  under local Alice and Bob spin-$1$
measurements do not, in general, monotonically increase with the growth of entanglement of a pure two-qutrit state. 

The article is organized as follows.

In Section \ref{preliminaries}, we specify the general expressions
(\ref{0_7}) and (\ref{0__9}) for the case of local Alice and Bob spin-$1$ measurements in a two-qutrit state
$\rho_{3\times3}$.

In Section \ref{spin_correlation_matrix_section}, for a variety of two-qutrit states, pure and mixed, we calculate analytically
the values of 
\emph{the spin-$1$ CHSH parameter} and show that, under local Alice and Bob spin-$1$
measurements in either of these states,  the CHSH inequality is not violated.

In Section \ref{entanglement_vs_nonviolation_chsh}, for all pure two-qutrit
states considered in Section \ref{spin_correlation_matrix_section}, we find
the values of \emph{their concurrence} and compare them with the values of
\emph{the spin-$1$ CHSH parameter } of these states. This allows us to show that, in contrast to the situation under spin-$\frac{1}{2}$ measurements in a pure two-qubit state, described by relation (\ref{0_5}), under local
Alice and Bob spin-$1$ measurements in a pure two-qutrit state,
there is no a monotonic dependence of the spin-$1$ CHSH parameter of this
pure state on its entanglement.

In Section 5, based on our analytical results in Section
\ref{spin_correlation_matrix_section} and the numerical results presented in
Appendix B on the calculation of the spin-$1$ CHSH parameter $\gamma_{s=1}(\rho_{3\times3})$ for more than
1,000,000 randomly generated pure nonseparable two-qutrit states, we put forward the
Conjecture that, under local spin-$1$ measurements in an arbitrary two-qutrit
state, pure or mixed, the CHSH inequality is not violated. 

In Section 6, we summarize the main results of the present article.

\section{The CHSH expectation under spin-$1$ measurements in a two-qutrit state}
\label{preliminaries}
For our further analysis in Sections 3--5, let us specify the general results (\ref{0_7}) and (\ref{0__9}) on Alice and Bob spin $s\geq\frac{1}{2}$ measurements, derived in \cite{arXiv}, for the case of spin-$1$ measurements. 

For the spin-$1$ qutrit, any spin-$1$ observable on $\mathbb{C}^{3}$ has the form
\begin{align}
S_{r}  &  =r\cdot S,\text{ \ \ }r\cdot S=\sum_{j=1,2,3}r_{j}S_{j}%
,\label{1_0}\\
r=(r_{1},r_{2},r_{3})  &  \in\mathbb{R}^{3},\text{\ \ \ }\left\Vert
r\right\Vert _{\mathbb{R}^{3}}=1\ ,\nonumber
\end{align}
and constitutes the projection on a direction $r\in\mathbb{R}^{3}$ of the
qutrit spin
\begin{equation}
S=(S_{1},S_{2},S_{3}),\ \text{\ \ }S^{2}=S_{1}^{2}+S_{2}^{2}+S_{3}%
^{2}=2\mathbb{I}_{\mathbb{C}^{3}}, \label{1_1}%
\end{equation}
where 
\begin{align}
\left[  S_{j},S_{k}\right]   &  =i\sum\varepsilon_{jkl}S_{l},\text{
\ \ }j,k,l=1,2,3,\label{1_3}\\
\mathrm{tr}[S_{j}S_{k}]  &  =2\delta_{jk}.\nonumber
\end{align}
Here, $\varepsilon_{jkl}$ is the Levi-Civita symbol and the Hermitian operators $S_i$, $i=1,2,3$, on  $\mathbb{C}^{3}$  are given by the expressions \begin{align}
S_{1}  &  =\frac{1}{\sqrt{2}}\sum_{n=1,2}\left(  |n\rangle\langle n+1|\text{
}+|n+1\rangle\langle n|\right)  ,\label{1_2}\\
S_{2}  &  =-\frac{i}{\sqrt{2}}\sum_{n=1,2}\left(  |n\rangle\langle n+1|\text{
}-|n+1\rangle\langle n|\right)  ,\nonumber\\
S_{3}  &  =|1\rangle\langle1|\text{ }-|3\rangle\langle3|,\nonumber
\end{align}
where $\{|n\rangle,$ $n=1,2,3\}$  is  the computational basis in
$\mathbb{C}^{3}$. In this basis, the matrix representation of any spin-$1$ observable
(\ref{1_0}) is given by
\begin{equation}
S_{r}:=%
\begin{pmatrix}
r_{3} & \frac{r_{1}-ir_{2}}{\sqrt{2}} & 0\\
\frac{r_{1}+ir_{2}}{\sqrt{2}} & 0 & \frac{r_{1}-ir_{2}}{\sqrt{2}}\\
0 & \frac{r_{1}+ir_{2}}{\sqrt{2}} & -r_{3}%
\end{pmatrix}
. \label{1_4}%
\end{equation}
Every spin-$1$ observable (\ref{1_0}) on $\mathbb{C}^{3}$ has the nondegenerate eigenvalues $\{-1,0,1$$\}.$

For Alice and Bob local spin-$1$ observables (\ref{1_0}), the CHSH expectation (\ref{chsh_expectation}) in a state $ \rho_{3\times3}$ is given by 
\begin{equation}
\langle\mathcal{B}_{chsh}(a_{1},a_{2};b_{1},b_{2})\rangle_{\rho_{3\times3}%
}=\text{ }\mathrm{tr}[\rho_{3\times3}\{S_{a_{1}}\otimes(S_{b_{1}} + S_{b_{2}} )\}]
+\mathrm{tr}[\rho_{3\times3}\{S_{a_{2}}\otimes(S_{b_{1}}- S_{b_{2}} )\}]\label{1_6},%
\end{equation}
where $a_{i}, b_{j}$  are unit vectors in $\mathbb{R}^{3}$. 

By specifying for the case of spin-$1$ measurements the general results (\ref{0_7}) and (\ref{0__9}), derived in 
Theorem 1 in \cite{arXiv} and true  for any spin 
 $s\geq\frac{1}{2}$, we
formulate for our further calculations the  following statements on the two-qutrit case.

\begin{proposition}
\label{proposition1}
For an arbitrary two-qutrit state $\rho_{3\times3}$, the maximum $\Upsilon_{chsh}^{(spin-1)}(\rho_{3\times3})$ of the absolute value of the CHSH
expectation (\ref{1_6}) over all Alice and Bob spin observables (\ref{1_0}) is given by
\begin{equation}
\Upsilon_{chsh}^{(spin-1)}(\rho_{3\times3})=\max_{\substack{a_{m},b_{k}%
\in\mathbb{R}^{3},\\\left\Vert a_{m}\right\Vert ,\left\Vert b_{k}\right\Vert
=1}}\left\vert \left\langle \mathcal{B}_{chsh}(a_{1},a_{2};b_{1}%
,b_{2})\right\rangle _{\rho_{3\times3}}\right\vert =2\sqrt{z^{2}(\rho
_{3\times3})+\widetilde{z}^{2}(\rho_{3\times3})}\ ,\label{1_10}%
\end{equation}
where $z(\rho_{3\times3})$ and $\widetilde{z}(\rho_{3\times3})$ are two
largest singular values of the spin-$1$ correlation matrix 
\begin{equation}
\mathcal{Z}^{(ij)}_{s=1}(\rho_{3\times 3}):=\mathrm{tr}[\rho_{3\times3}\{S_{i}\otimes
S_{j}\}]\in\mathbb{R},\text{ \ \ }i,j=1,2,3,\label{1_8}%
\end{equation}
defined via the spin-$1$ components in (\ref{1_2}).
\end{proposition}

By the relation (21) in \cite{arXiv} the
operator norm of the spin-$1$ correlation matrix
\begin{equation}
\left\Vert \mathcal{Z}_{s=1}(\rho_{3\times 3})\right\Vert _{0}\leq 1\ ,\label{1_9}%
\end{equation}
so that its singular values cannot exceed $1$. Proposition \Ref{proposition1} implies the following corollary.

\begin{corollary}
\label{corollary1}
For a two-qutrit state $\rho_{3\times3},$ the ratio of the maximum (\ref{1_10}) of the absolute value of the
CHSH expectation (\ref{1_6}) under spin-$1$ measurements to the CHSH 
maximum in an LHV case is given by
\begin{equation}
\gamma_{s=1}(\rho_{3\times3})=\sqrt{z^{2}(\rho_{3\times3})+\widetilde{z}^{2}%
(\rho_{3\times3})}\label{1_11}%
\end{equation}
and, in view of (\ref{1_9}), is upper bounded by the Tsirelson \cite{7,8}
bound $\sqrt{2}$. A two-qutrit state $\rho_{3\times3}$ violates the CHSH
inequality under Alice and Bob spin-$1$ measurements if and only if
its spin-$1$ CHSH parameter satisfies the condition  $\gamma_{s=1}(\rho_{3\times 3})>1.$
\end{corollary}

If a two-qutrit state $\rho_{3\times3}$, on $\mathbb{C}^{3}\otimes\mathbb{C}^{3}$ is
given by a convex combination of some two-qutrit states $\rho_{3\times 3}^{(l)}$ , that is:
\begin{equation}
\rho_{3\times3}=\sum_{l}\\
\xi_{l}\rho_{3\times3}^{(l)}, \text\
\xi_{l}>0, \\\sum_{l}\xi_{l}=1,
\end{equation}
then
\begin{align}
\gamma_{s=1}(\rho_{3\times3}) &  =\max_{\substack{a_{m},b_{k}\in\mathbb{R}%
^{3},\\\left\Vert a_{m}\right\Vert ,\left\Vert b_{k}\right\Vert =1}}\left\vert
\left\langle\mathcal{B}_{chsh}(a_{1},a_{2};b_{1},b_{2})\right\rangle
_{\rho_{3\times3}}\right\vert \label{1_12}\\
&  \leq\sum_{l}\xi_{l}\max_{\substack{a_{m},b_{k}\in\mathbb{R}^{3}%
,\\\left\Vert a_{m}\right\Vert ,\left\Vert b_{k}\right\Vert =1}}\left\vert
\left\langle \mathcal{B}_{chsh}(a_{1},a_{2};b_{1},b_{2})\right\rangle
_{\rho_{3\times3}^{(l)}}\right\vert =\sum_{l}\xi_{l}\text{ }\gamma_{s=1}
(\rho_{3\times3}^{(l)}).\nonumber
\end{align}

Note that, due to the
algebraic inequality $|x\pm y|$ $\leq1\pm xy,$ valid for all $x,y\in\lbrack-1,1]$ and relation $\mathrm{tr}
{[\rho\ S_{r}]}\leq1$, which holds for arbitrary states $\rho$ and spin-$1$ observables (\ref{1_0}), for any separable two-qutrit state 
\begin{equation}
\rho_{3\times3}%
^{(sep)}=\sum_{l}\xi_{l}\rho_{3}^{(1,l)}\otimes\rho_{3}^{(2,l)},
\end{equation}
where
$\rho_{3}^{(j,l)},$ $j=1,2,$ are states on $\mathbb{C}^{3},$ the CHSH expectation (\ref{1_6}) in a factorized state
$\rho_{3}^{(1,l)}\otimes\rho_{3}^{(2,l)}$ satisfies the relations
\begin{align}
&  \left\vert \langle\mathcal{B}_{chsh}(a_{1},a_{2};b_{1},b_{2})\rangle
_{\rho_{3}^{(1,l)}\otimes\rho_{3}^{(2,l)}}\right\vert \label{1_13}\\
&  \leq\left\vert \lbrack\mathrm{tr}[\rho_{3}^{(2,l)}S_{b_{1}}]+\mathrm{tr}%
[\rho_{3}^{(2,l)}S_{b_{2}}]\right\vert \text{ }+\left\vert [\mathrm{tr}%
[\rho_{3}^{(2,l)}S_{b_{1}}]-\mathrm{tr}[\rho_{3}^{(2,l)}S_{b_{2}}]\right\vert
\nonumber\\
&  \leq2,\nonumber
\end{align}
so that, for any factorized state $\rho_{3}^{(1,l)}\otimes\rho_{3}^{(2,l)}$,
the spin-$1$ CHSH parameter $\gamma_{s=1}(\rho_{3}^{(1,l)}\otimes\rho_{3}
^{(2,l)})\leq1$. Taking this into account in relation (\ref{1_12}) we have:
\begin{equation}
\label{separ}
\gamma_{s=1}\left(
\rho_{3\times3}^{(sep)}\right)  \leq 1\ . 
\end{equation} 
Thus, the general relation (\ref{1_12}) incorporates the well-known fact that a separable state does not violate the CHSH inequality.

With respect to the spin-$1$ correlation matrix Corollary \ref{corollary1} and relation (\ref{separ}) imply.
\begin{corollary}
\label{corollary2}
For every separable two-qutrit state the sum of the two largest singular values of the spin-$1$ correlation matrix satisfies the relation
\begin{equation}
z^{2}(\rho
_{3\times3})+\widetilde{z}^{2}(\rho_{3\times3})\leq 1\ .
\end{equation}

\end{corollary}

In the following sections, based on Proposition \ref{proposition1} and Corollary
\ref{corollary1}, we analyze the value of the spin-$1$ CHSH parameter (\ref{1_11}) for a variety of nonseparable two-qutrit
states.

\section{Spin correlation matrix for a two-qutrit state}
\label{spin_correlation_matrix_section}
In this Section, we find the spin
correlation matrix (\ref{1_8}) for an arbitrary two-qutrit state $\rho_{3\times 3}$ on
$\mathbb{C}^{3}\otimes\mathbb{C}^{3}$.
Let 
\begin{align}
\rho_{3\times 3}  &  =\sum\zeta_{mm^{\prime},kk^{\prime}}|mk\rangle\langle
m^{\prime}k^{\prime}|,\label{22_}\\
\zeta_{mm^{\prime},kk^{\prime}}  &  =\langle mk|\rho_{3\times 3}|m^{\prime
}k^{\prime}\rangle ,\nonumber\\
\zeta_{mm^{\prime},kk^{\prime}}^{\ast}  &  =\zeta_{m^{\prime}m,k^{\prime}%
k},\text{ \ \ }\sum_{m,k}\zeta_{mm,kk}=1, \nonumber
\end{align}
be the representation of a state  $\rho_{3\times 3}$ via the elements in the computational basis of $\mathbb{C}^{3}\otimes\mathbb{C}^{3}.$

Specifying relations (44)-(46) in \cite{arXiv} for the spin-$1$ correlation matrix of a two-qutrit state $\rho_{3\times 3}$, we come to the following results. 
The elements of the first row:
\begin{align}
\mathcal{Z}_{s=1}^{(11)}(\rho_{3\times 3})  &  =\frac{1}{2}\sum_{m,k=1}%
^{2}\text{ }\sqrt{mk(3-m)(3-k)}\mathrm{Re}\left[  \text{ }\zeta
_{m(m+1),k(k+1)}+\zeta_{m(m+1),(k+1)k}\right]  ,\label{z1}\\
\mathcal{Z}_{s=1}^{(12)}(\rho_{3\times 3})  &  =\frac{1}{2}\sum_{m,k=1}%
^{2}\text{ }\sqrt{mk(3-m)(3-k)}\mathrm{Im}\left[  \zeta_{m(m+1),k(k+1)}%
+\zeta_{(m+1)m,k(k+1)}\right]  ,\nonumber\\
\mathcal{Z}_{s=1}^{(13)}(\rho_{3\times 3})  &  =\frac{1}{2}\sum_{m,k=1}%
^{3}\text{ }\sqrt{m(3-m)}\left(  4-2k\right)  \mathrm{Re}\left[
\zeta_{(m+1)m,kk}\right]  .\nonumber
\end{align}
The elements of the second row:
\begin{align}
\mathcal{Z}_{s=1}^{(21)}(\rho_{3\times 3})  &  =\frac{1}{2}\sum_{m,k=1}%
^{2}\text{ }\sqrt{mk(3-m)(3-k)}\text{ }\mathrm{Im}\left[  \zeta
_{m(m+1),k(k+1)}+\zeta_{m(m+1),(k+1)k}\right]  ,\label{z2}\\
\mathcal{Z}_{s=1}^{(22)}(\rho_{3\times 3})  &  =\frac{1}{2}\sum_{m,k=1}%
^{2}\text{ }\sqrt{mk(3-m)(3-k)}\text{ }\mathrm{Re}\left[  \zeta
_{(m+1)m,k(k+1)}-\zeta_{(m+1)m,(k+1)k}\right]  ,\nonumber\\
\mathcal{Z}_{s=1}^{(23)}(\rho_{3\times 3})  &  =\frac{1}{2}\sum_{m,k=1}%
^{3}\text{ }\sqrt{m(3-m)}\left(  4-2k\right)  \mathrm{Im}\left[
\zeta_{m(m+1),kk}\right] ,\nonumber
\end{align}
and the elements of the third row:
\begin{align}
\mathcal{Z}_{s=1}^{(31)}(\rho_{3\times 3})  &  =\frac{1}{2}\sum_{m,k=1}%
^{3}\text{ }\left(  4-2m\right)  \sqrt{k(3-k)}\mathrm{Re}\left[
\zeta_{mm,(k+1)k}\right]  ,\label{z3}\\
\mathcal{Z}_{s=1}^{(32)}(\rho_{3\times 3})  &  =\frac{1}{2}\sum_{m,k=1}%
^{3}\text{ }\left(  4-2m\right)  \sqrt{k(3-k)}\mathrm{Im}\left[
\zeta_{mm,k(k+1)}\right]  ,\nonumber\\
\mathcal{Z}_{s=1}^{(33)}(\rho_{3\times 3})  &  =\frac{1}{4}\sum_{m,k=1}%
^{3}(4-2m)(4-2k)\zeta_{mm,kk}\ .\ \nonumber
\end{align}

In the following subsections, using  relations (\ref{z1})--(\ref{z3}), we find the values of the spin-$1$ CHSH parameter (\ref{1_11}) for a variety of pure and mixed two-qutrit states. Recall that,  according to Corollary \ref{corollary1}, under local spin-$1$ measurements in  a two-qutrit state, the CHSH inequality is violated iff the value of this parameter is greater than $1$. 

\subsection{Pure two-qutrit states}

In this Section, we compute the values of the spin-$1$ CHSH parameter (\ref{1_11}) for a variety of pure two-qutrit states. 

Consider first the family of pure two-qutrit states  $|\psi_{3\times3}^{(asym)}\rangle\langle
\psi_{3\times3}^{(asym)}|$, where the unit vector $|\psi_{3\times3}^{(asym)}\rangle$ belongs to the subspace of antisymmetric vectors in $\mathbb{C}^{3} \otimes\mathbb{C}^{3}$. 
In this $3$-dimensional
subspace, the following three 
antisymmetric unit vectors
\begin{align}
|\phi_{12}^{(-)}\rangle & = \frac{1}{\sqrt{2}}\left( |1\rangle\otimes
|2\rangle-|2\rangle\otimes|1\rangle\right) ,\\
|\phi_{13}^{(-)}\rangle & = \frac{1}{\sqrt{2}}\left( |1\rangle\otimes
|3\rangle-|3\rangle\otimes|1\rangle\right) ,\nonumber\\
|\phi_{23}^{(-)}\rangle & = \frac{1}{\sqrt{2}}\left( |2\rangle\otimes
|3\rangle-|3\rangle\otimes|2\rangle\right) ,\nonumber
\end{align}
constitute an orthonormal basis, so that the decomposition of each vector $|\psi_{3\times3}^{(asym)}\rangle$ in  $\mathbb{C}^{3}\otimes\mathbb{C}^{3}$ reads
\begin{align}
\label{anti}|\psi_{3\times3}^{(asym)}\rangle & =\alpha_{12} |\phi_{12}%
^{(-)}\rangle+\alpha_{13}|\phi_{13}^{(-)}\rangle+\alpha_{23}|\phi_{23}%
^{(-)}\rangle\ ,\\
\alpha_{12},\alpha_{13} & ,\alpha_{23}\in\mathbb{C},\ \ \ |\alpha_{12}%
|^{2}+|\alpha_{13}|^{2} + |\alpha_{23}|^{2}=1\ .\ \nonumber
\end{align}
Each pure state  $|\psi_{3\times3}^{(asym)}\rangle\langle\psi_{3\times
3}^{(asym)}|$ of the form (\ref{anti}) is nonseparable, see Proposition \ref{proposition2} in Section \ref{entanglement_vs_nonviolation_chsh}. 

In decomposition (\ref{22_}), the non-vanishing coefficients of this pure state 
read
\begin{equation}
\zeta_{ii,jj}=-\zeta_{ij,ji}=-\zeta_{ji,ij}=\zeta_{jj,ii}=\frac{|\alpha
_{ij}|^{2}}{2}\ ,\text{ \ \ \ }i\neq j,\text{ \ }i,j=1,2,3,\label{31}%
\end{equation}
and
\begin{align}
\zeta_{12,j3} &  =-\zeta_{13,j2}=\zeta_{21,3j}^{\ast}=-\zeta_{j2,13}%
=-\zeta_{2j,31}^{\ast}=\zeta_{j3,12}=-\zeta_{31,2j}^{\ast}=\zeta_{3j,21}%
^{\ast}=\frac{\alpha_{1j}\alpha_{23}^{\ast}}{2},\label{31_1}\ \\
&j=2,3.\nonumber
\end{align}

Note that since state  $|\psi_{3\times3}^{(asym)}\rangle\langle\psi_{3\times3}^{(asym)}|$ is
invariant under the permutation of the Hilbert spaces in the tensor product
$\mathbb{C}^{3}\otimes\mathbb{C}^{3}$, by (\ref{1_8}) its spin-$1$ correlation
matrix $\mathcal{Z}_{s=1}(|\psi_{3\times3}^{(asym)}\rangle)$ is symmetric.

From relations (\ref{z1})--(\ref{z3}) and (\ref{31}), (\ref{31_1}) it follows that this matrix has the form
\begin{equation}
\label{61}\frac{1}{2}\left(
\begin{array}
[c]{ccc}%
-|\alpha_{12}-\alpha_{23}|^{2} & 2 \mathrm{Im}\left( \alpha_{12}^{*}%
\alpha_{23}\right)  & \sqrt{2} \mathrm{Re}\left( \alpha_{12}\alpha_{13}%
^{*}-\alpha_{13}\alpha_{23}^{*}\right) \\
2 \mathrm{Im}\left( \alpha_{12}^{*}\alpha_{23}\right)  & -|\alpha_{12}%
+\alpha_{23}|^{2} & \sqrt{2}\mathrm{Im}\left( \alpha_{13}(\alpha_{12}%
^{*}+\alpha_{23} ^{*})\right) \\
\sqrt{2}\mathrm{Re}\left( \alpha_{12}\alpha_{13}^{*}-\alpha_{13}\alpha
_{23}^{*}\right)  & \sqrt{2}\mathrm{Im}\left( \alpha_{13}(\alpha_{12}%
^{*}+\alpha_{23} ^{*})\right)  & -2 |\alpha_{13}|^{2}%
\end{array}
\right) \
\end{equation}
The singular values of this matrix are equal to $0$ and 
$\frac{1}{2}|(1\pm|\alpha_{13}^{2}-2\alpha_{12}\alpha_{23}|)|$. 

Therefore, by (\ref{1_11}), 
for any pure state with a vector $|\psi^{(asym)}_{3\times3}\rangle$ of the form  (\ref{anti}), the spin-$1$ CHSH parameter is given by
\begin{equation}
\label{gen_ant}\gamma_{s=1}(|\psi_{3\times3}^{(asym)}\rangle)=\sqrt
{z^2(|\psi_{3\times3}^{(asym)}\rangle)+\tilde{z}^2(|\psi_{3\times3}^{(asym)}\rangle)}=
\sqrt{\frac{1+|\alpha_{13}^{2}-2 \alpha_{12} \alpha_{23}|^{2}}{2}}\ .
\end{equation}
Taking into account the normalization relation in (\ref{anti}), for the
radicand in (\ref{gen_ant}), we have
\begin{align}
\label{63}\frac{1+|\alpha_{13}^{2}-2 \alpha_{12} \alpha_{23}|^{2}}{2} 
&\leq \frac{1+\left(|\alpha_{13}|^{2}+2 |\alpha_{12} \alpha_{23}|\right)^{2}}{2} \\
& \leq\frac{1+(|\alpha_{13} |^{2}+|\alpha_{12}|^{2}+|\alpha_{23}|^{2})^2}{2}
=1\ .\nonumber
\end{align}
From Eqs.  (\ref{gen_ant}) and (\ref{63}) it follows that, for a pure
two-qutrit state with vector  
$|\psi_{3\times3}^{(asym)}\rangle$, the spin-$1$ CHSH
parameter 
\begin{equation}
1/2\leq \gamma_{s=1}(|\psi_{3\times3}^{(asym)}\rangle)\leq 1\ .\label{x} 
\end{equation}
Therefore,  by Corollary \ref{corollary1}, under Alice and Bob spin-$1$ measurements in each of nonseparable pure two-qutrit
states (\ref{anti}) the CHSH inequality is not violated.

Let us further consider the family of pure two-qutrit  states  $|\psi_{3\times3}^{(sym)}\rangle\langle
\psi_{3\times3}^{(sym)}|$, where a unit vector $|\psi_{3\times3}^{(sym)}\rangle$ belongs to the symmetric subspace of  $\mathbb{C}^{3} \otimes\mathbb{C}^{3}$, moreover, has the form 
\begin{align}
\label{symm_qutrits}|\psi_{3\times3}^{(sym)}\rangle=\alpha_{11}|\phi
_{11}^{(+)}\rangle+\alpha_{22}|\phi_{22}^{(+)}\rangle+\alpha_{33}|\phi
_{33}^{(+)}\rangle\ ,\\\
\ \ \alpha_{11},\alpha_{22},\alpha_{33}\in\mathbb{C}, \ \ |\alpha_{11}%
|^{2}+|\alpha_{22}|^{2} + |\alpha_{33}|^{2}=1\ ,\nonumber
\end{align}
where $|\phi_{ii}^{(+)}\rangle  = |ii\rangle ,\  i= 1,2,3$, are mutually orthogonal symmetric unit vectors in $\mathbb{C}^{3}\otimes\mathbb{C}^{3}$. Each of the pure states of the form (\ref{symm_qutrits}) is nonseparable unless any two of the coefficients $\alpha_{ii}$ are simultaneously equal to zero (see Proposition \ref{proposition3} in Section \ref{entanglement_vs_nonviolation_chsh}).

For the state $|\psi_{3\times3}%
^{(sym)}\rangle\langle\psi_{3\times3}^{(sym)}|$, the coefficients in decomposition 
(\ref{22_}) are given by
\begin{equation}
\label{coefsym}
\zeta_{ij,ij}=\zeta_{ji,ji}^{*}=\alpha_{ii}\alpha_{jj}^{*} ,\   i,j = 1,2,3,\  i\leq j , 
\end{equation}
so that by Eqs. (\ref{z1})--(\ref{z3}) and (\ref{coefsym}), for
a pure two-qutrit state of this type, the spin-$1$ correlation matrix is equal to
\begin{equation}
\mathcal{Z}_{s=1}(|\psi_{3\times3}^{(sym)}\rangle)=
\begin{pmatrix}
\mathrm{Re}\left(  \alpha_{11}^{*}\alpha_{22}+\alpha_{22}^{*}\alpha
_{33}\right)  & \mathrm{Im}\left(  \alpha_{11}^{*}\alpha_{22}+\alpha_{22}%
^{*}\alpha_{33}\right)  & 0\\
\mathrm{Im}\left(  \alpha_{11}^{*}\alpha_{22}+\alpha_{22}^{*}\alpha
_{33}\right)  & -\mathrm{Re}\left(  \alpha_{11}^{*}\alpha_{22}+\alpha_{22}%
^{*}\alpha_{33}\right)  & 0\\
0 & 0 & |\alpha_{11}|^{2} +|\alpha_{33}|^{2}\\
&  &
\end{pmatrix}\ ,
\end{equation}
 and has the singular value $
|\alpha_{11}|^{2}+|\alpha_{33} |^{2}$ of multiplicity one and the singular value 
$|\alpha_{11}^{*}\alpha_{22}+\alpha_{22}^{*}\alpha_{33}|$
of multiplicity two.

Consider first the case where $|\psi_{3\times3}^{(sym)}\rangle$ is such that, in decomposition (\ref{symm_qutrits}), its coefficients satisfy the relation
\begin{equation}
|\alpha_{11}^{*}\alpha_{22}+\alpha_{22}^{*}%
\alpha_{33}| \geq|\alpha_{11}|^{2}+|\alpha_{33}|^{2}.\label{c} \end{equation}
In this case, by (\ref{1_11}) and the above singular values of matrix $\mathcal{Z}_{s=1}(|\psi_{3\times3}^{(sym)}\rangle)$ the spin-$1$ CHSH parameter for the corresponding pure two-qutrit state is given by 
\begin{equation}
\gamma_{s=1}(|\psi_{3\times3}
^{(sym)}\rangle)=\sqrt{2}|\alpha_{11}^{*}\alpha_{22}+\alpha_{22}^{*}\alpha_{33}|\label{b}\ .
\end{equation}
Taking into account that, in view of the normalization condition
\begin{align}
\label{73}2 |\alpha_{11}^{*}\alpha_{22}+\alpha_{22}^{*}\alpha_{33}|^{2}  &
\leq2 |\alpha_{22}|^{2} \left(  |\alpha_{11}|^{2}+|\alpha_{33}|^{2}%
+2|\alpha_{11}||\alpha_{33}|\right) \\
& \leq4 |\alpha_{22}|^{2} \left(  |\alpha_{11}|^{2}+|\alpha_{33}|^{2}\right)
\nonumber\\
& \leq4 |\alpha_{22}|^{2} \left(  1-|\alpha_{22} |^{2}\right)  \ \nonumber
\end{align}
and relation $|\alpha_{22}|^{2} \left(
1-|\alpha_{22} |^{2}\right) \leq1/4$, we conclude that, for the pure state $|\psi_{3\times3}
^{(sym)}\rangle$ in case (\ref{c}), the spin-$1$ CHSH parameter (\ref{b}) is upper bounded
by 
\begin{equation}
\label{gamma_sym_1}\gamma_{s=1}(|\psi_{3\times3}^{(sym)}\rangle)\leq 1\ .
\end{equation}

In the case opposite to (\ref{c}):  
\begin{equation}
|\alpha_{11}^{*}\alpha_{22}+\alpha_{22}^{*}\alpha_{33}|
<|\alpha_{11} |^{2}+|\alpha_{33}|^{2},\label{d}
\end{equation} 
the spin-$1$ CHSH parameter of state $|\psi_{3\times3}
^{(sym)}\rangle$ is equal to 
\begin{equation}
\label{gamma_sym}
\gamma_{s=1}(|\psi_{3\times3}^{(sym)}\rangle)= 
\sqrt{|\alpha_{11}^{*}\alpha_{22}+\alpha_{22}^{*}\alpha_{33}|^{2} +\left(
|\alpha_{11}|^{2}+|\alpha_{33}|^{2}\right)^{2}} \ .
\end{equation}
Taking here into account that
\begin{align}
|\alpha_{11}^{*}\alpha_{22}+\alpha_{22}^{*}\alpha_{33}|^{2} +\left(
|\alpha_{11}|^{2}+|\alpha_{33} |^{2}\right)  ^{2}  &  \leq|\alpha_{22}%
|^{2}\left(  |\alpha_{11}|+|\alpha_{33}|\right)  ^{2} +(1-|\alpha_{22}
|^{2})^{2}\\
&\leq  2|\alpha_{22}|^{2}(1-|\alpha_{22}|^{2})+(1-|\alpha_{22}
|^{2})^{2}\nonumber\\
&  =1-|\alpha_{22}|^{4}\ ,\nonumber
\end{align}
we conclude that, as in the previous case, 
\begin{equation}
\label{gamma_sym_2}
\gamma_{s=1}(|\psi_{3\times
3}^{(sym)}\rangle)\leq 1\ .
\end{equation}
Consequently, under local Alice and Bob spin-$1$ measurements in a pure symmetric
two-qutrit state of the form (\ref{symm_qutrits}), the CHSH inequality is not violated.

\subsubsection{Two-qutrit GHZ state}
The GHZ state on $\mathbb{C}^3\otimes \mathbb{C}^3$ is a particular case of pure states (\ref{symm_qutrits}), namely, for the two--qutrit GHZ state, vector $|\psi_{3\times
3}^{(sym)}\rangle$ in (\ref{symm_qutrits}) is given as $|GHZ_{3}\rangle= \frac{1}{\sqrt{3}} \sum_{m=1,2,3} |mm\rangle$.
This and Eq. (\ref{b}) imply that the spin-$1$ CHSH parameter of the GHZ state is equal to
\begin{equation}
\label{ghz_gamma}
\gamma_{s=1}(|GHZ_{3}\rangle) = \sqrt{\frac{8}{9}}\ .
\end{equation}

We stress that, for a pure separable two-qutrit state, let described by coefficients  $\alpha_{11}=1$ and $\alpha_{22}=\alpha_{33}=0$ in Eq. (\ref{symm_qutrits}), the spin-$1$ CHSH parameter is equal to unity. 

This and Eq. (\ref{ghz_gamma}) imply that, for a maximally entangled two-qutrit state,  the spin-$1$ CHSH parameter may be less than that for a separable pure two-qutrit state.

\subsection{Mixed two-qutrit states}
\label{mixed_twoqutrit_states}

In this section, we calculate the values of the spin-$1$ CHSH parameter for the Werner state \cite{Wer:89} and for the Horodecki state \cite{Hor.Hor.Hor:99}.

\subsubsection{Two-qutrit Werner state}
The two-qutrit Werner state is defined \cite{Wer:89} as
\begin{align}
\rho _{3,\Phi }^{(wer)}=\frac{3-\Phi }{24}&\mathbb{I}_{\mathbb{C}%
^{3}\otimes \mathbb{C}^{3}}+\frac{3\Phi -1}{24}V_{3}\label{g} ,\ \ \text{}\
\Phi \ \in \left[ -1,1\right] ,
\end{align}
where $V_{3}(\psi_{1}\otimes \psi_{2}):=\psi_{2}\otimes \psi_{1}$ is the
permutation operator on $\mathbb{C}^{3}\otimes \mathbb{C}^{3}$. The state in (\ref{g}) is separable iff $\Phi \ \in \ \left[ 0,1\right] $ and nonseparable otherwise, also, under projective measurements of Alice
and Bob, the nonseparable Werner state $\rho_{3,\Phi }^{(wer)}$ admits an LHV model for all $\Phi \in \lbrack -\frac{5}{9},0).$

From Eq. (57) in \cite{arXiv} it follows that the spin-$1$ CHSH parameter of the two-qutrit Werner state is given by

\begin{equation}
\gamma_{s=1}(\rho_{3,\Phi }^{(wer)}) =\frac{\sqrt{2}}{12}|3\Phi -1|\leq 1\ ,\ \forall\ \Phi\in \left[ -1,1\right]\ . \label {yy}
\end{equation}
This and Corollary \ref{corollary1} imply that, under local Alice and Bob spin-$1$ measurements in any nonseparable Werner state  $\rho_{3,\Phi }^{(wer)}$, even nonlocal ($ -1 \leq \Phi <-\frac{5}{9}$), the CHSH inequality is not violated\footnote{This is consistent with a more general result following from Theorem 3 in \cite{Lou:05} that, for any $d>2$, the nonseparable Werner state does not violate the CHSH inequality under all types of local Alice and Bob measurements. 
}.

\subsubsection{Two-qutrit Horodecki state}
Consider the Horodecki mixed state
\begin{equation}
\rho^{(hor)}_{3\times3}(\tau)=\frac{2}{7}\ \rho_{3\times3}%
^{(ghz)}+\frac{\tau}{7}\ \xi_{3\times3}^{(+)}+\frac{5-\tau}{7}%
\xi_{3\times3}^{(-)}\ ,\text\ \tau\in\left[ 2,5\right],\label{hor}
\end{equation}
introduced in \cite{Hor.Hor.Hor:99}. Here: (i) $\rho_{3\times 3}^{(ghz)}=\frac{1}{3}\sum_{m,m^{\prime}}|mm\rangle\langle
m^{\prime}m^{\prime}|$ is the
two-qutrit Greenberger--Horne--Zeilinger (GHZ) state, which is maximally entangled, and (ii) $\xi_{3\times3}^{(\pm)}$ are the mixed separable states
\begin{align}
\xi_{3\times3}^{(+)} & =\frac{1}{3}\left(  |12\rangle\langle12|+|23\rangle
\langle23|+|31\rangle\langle31| \right) \ ,\\
\xi_{3\times3}^{(-)} & =\frac{1}{3}\left(  |21\rangle\langle 21|+|32\rangle
\langle 32|+|13\rangle\langle 13| \right) \ .\nonumber
\end{align}

As it is proved in \cite{Hor.Hor.Hor:99}, the Horodecki state (\ref{hor}) is
separable if $2\leq\tau\leq3$, bound entangled if $3<\tau\leq4$ and free
entangled for $4<\tau\leq5$.

For the Horodecki state (\ref{hor}), the nonzero coefficients in decomposition (\ref{22_}) are given
by
\begin{align}
\label{coef_hor} & \zeta_{11,11}=\zeta_{12,12}=\zeta_{13,13}= \zeta
_{21,21}=\zeta_{22,22}=\zeta_{23,23}= \zeta_{31,31}=\zeta_{32,32}%
=\zeta_{33,33}=\frac{2}{21}\ ,\\
& \zeta_{11,22}=\zeta_{22,33}=\zeta_{33,11}=\frac{\tau}{21}\ ,\ \zeta
_{11,33}=\zeta_{22,11}=\zeta_{33,22}=\frac{5-\tau}{21}\ ,\nonumber
\end{align}
so that by Eqs. (\ref{z1})-(\ref{z3}) the spin correlation matrix for this state has the form
\begin{equation}
\mathcal{Z}_{s=1}( \rho^{(hor)}_{3\times3}(\tau))=\frac{1}{21}
\begin{pmatrix}
4 & 0 & 0\\
0 & -4 & 0\\
0 & 0 & -1\\
\end{pmatrix}
\ ,
\end{equation}
for all $\tau\in\left[ 2,5\right]$. This matrix has the greatest singular value $4/21$ of multiplicity 2, so that the spin-$1$ CHSH parameter (\ref{1_11}) is equal to 
\begin{equation}
\label{gamma_hor}\gamma_{s=1}( \rho^{(hor)}_{3\times3}(\tau))=\frac{4\sqrt{2}%
}{21}%
\end{equation}
and does not depend on the value of parameter $\tau\in\left[ 2,5\right]$, which defines the
entanglement class of the Horodecki state. 

Since the spin-$1$ CHSH parameter (\ref{gamma_hor}) is less than one for all $\tau\in\left[ 2,5\right]$, by Corollary \ref{corollary1}, under
Alice and Bob spin-$1$ measurements in the Horodecki state (\ref{hor}),
the CHSH inequality is not violated independently of the entanglement class of this state. 

\section{Entanglement versus the CHSH nonviolation}
\label{entanglement_vs_nonviolation_chsh}

In this Section, we analyze the relation between values of the spin-$1$ CHSH parameter for pure two-qutrit states $|\psi_{3\times3}\rangle\langle
\psi_{3\times3} | $, considered in Sections 3.1  and 3.2,  and  values of their concurrence 
\cite{Che.Alb.Fei:05} 

\begin{equation}
\label{concurrence}\mathrm{C}(|\psi_{3\times3}\rangle)=\sqrt{2\left(
1-\mathrm{tr}\left[ \rho_{j}^{2}\right] \right) }\ .
\end{equation}
Here, $\rho_{j}$, $j=1,2$, are the states on
$\mathbb{C}^{3}$, reduced from $|\psi_{3\times3}\rangle\langle
\psi_{3\times3} | $,  and by the Schmidt theorem $\mathrm{tr}[\rho_{1}%
^{2}]=\mathrm{tr}[\rho_{2}^{2}]$. Note that for a maximally
entangled two-qutrit state $\mathrm{C}(|\psi_{3\times3}\rangle)=2/\sqrt{3}.$

\begin{proposition}
\label{proposition2}  For every pure two-qutrit state of the form (\ref{anti}), the concurrence (\ref{concurrence}) is equal to 
\begin{equation}
\mathrm{C}(|\psi_{3\times3}^{(asym)}\rangle)=1\ .\
\end{equation}
\end{proposition}
\begin{proof}
We find that, for every pure state of the form (\ref{anti}), the reduced states $\rho_j, j=1,2$, satisfy the relation\footnote{For the matrix representation of the density operator $|\psi^{(asym)}_{3\times 3}\rangle\langle\psi^{(asym)}_{3\times 3}|$ in the computational basis of $\mathbb{C}^3$, see Appendix A.}
\begin{equation}
\mathrm{tr}\left[  \rho_{j}^{2}\right] =\frac{1}{2}\left(  |\alpha_{12}%
|^{2}+|\alpha_{13}|^{2}+|\alpha_{23}|^{2}\right) ^{2}= 1/2\ ,
\end{equation}
so that by (\ref{concurrence})
\begin{equation}
\mathrm{C}(|\psi_{3\times3}^{(asym)}\rangle)=\sqrt{2(1- \mathrm{tr}\left[ \rho_{j}%
^{2}\right]) }= 1\ .
\end{equation} \end{proof}

By Proposition \ref{proposition2}, all pure two-qutrit states of the form (\ref{anti}) are
nonseparable, moreover, have the same value of the concurrence which is equal to one. 

Comparing this result with the values (\ref{gen_ant}) of the spin-$1$ CHSH parameter for these states, we find that, despite the same degree of entanglement for all pure states of the form (\ref{anti}),
the spin-$1$ CHSH parameter of these states varies in $\left[ \frac{1}{2},1\right] $.

For the concurrence of a pure two-qutrit state $|\psi_{3\times3}%
^{(sym)}\rangle\langle\psi_{3\times3}^{(sym)}|$ with the symmetric vector
$|\psi_{3\times3}^{(sym)}\rangle$ of the form (\ref{symm_qutrits}), we have the following result.

\begin{proposition}
\label{proposition3}
For every pure two-qutrit state with vector
$|\psi_{3\times3}^{(sym)}\rangle$ of the form (\ref{symm_qutrits}) its concurrence is equal to
\begin{equation}
\label{conc_sym}\mathrm{C}(|\psi_{3\times3}^{(sym)}\rangle)=\sqrt
{2(1-|\alpha_{11}|^{4}-|\alpha_{22}|^{4}-|\alpha_{33}|^{4})}\ .
\end{equation}
\end{proposition}
\begin{proof}
    The reduced states $\rho_j,\ j=1,2,$ of $
|\psi_{3\times3}^{(sym)}\rangle\langle\psi_{3\times3}^{(sym)}| $ have the form
\begin{align}
\rho_{j} 
& = |\alpha_{11}|^{2} |1\rangle\langle 1 | +|\alpha_{22}|^{2} |2\rangle
\langle 2 |+|\alpha_{33}|^{2} |3\rangle\langle 3 | \ .
\end{align}
This expression and relation (\ref{concurrence}) imply
\begin{equation}
\mathrm{C}(|\psi_{3\times3}^{(sym)}\rangle)=\sqrt{2(1- \mathrm{tr}\left[ \rho_{j}%
^{2}\right] )} =\sqrt{2\left( 1-|\alpha_{11}|^{4}-|\alpha_{22}|^{4}%
-|\alpha_{33}|^{4}\right) }\ .
\end{equation}
\end{proof}

Proposition \ref{proposition3} allows us to compare the values of the concurrence (\ref{conc_sym}) and the spin-$1$ CHSH parameter (\ref{gamma_sym}) of a pure two-qutrit symmetric state of the form (\ref{symm_qutrits}). 

For example, in case of the two-qutrit GHZ state $\rho_{3\times3}^{(ghz)}=|\psi_{3\times3}^{(ghz)}\rangle \langle \psi_{3\times3}^{(ghz)}|$, where all three coefficients in (\ref{symm_qutrits}) are equal to  $\alpha_{ii}=1/\sqrt
{3}$, $i=1,2,3$,  the concurrence (\ref{conc_sym}) attains the maximal value 
$\mathrm{C}(|\psi_{3\times3}^{(ghz)}\rangle)=2/\sqrt{3}$ among all two-qutrit states, while the spin-$1$ CHSH parameter of this state is equal, by Eq. (\ref{ghz_gamma}), to $\gamma_{s=1}(|\psi_{3\times3}^{(ghz)}\rangle
)=\sqrt{\frac{8}{9}}<1$. However, for any separable pure symmetric two-qutrit state $|\psi_{3\times3}^{(sep)}\rangle$, where the concurrence (\ref{conc_sym}) vanishes, let with the nonzero component $\alpha_{11}=1$, 
 the spin-$1$ 
CHSH parameter $\gamma_{s=1}(|\psi_{3\times3}^{(sep)}\rangle)=1$.

This, in particular, implies that, under local Alice and Bob spin-$1$ measurements, the absolute value of the CHSH expectation (\ref{1_6}) in a maximally entangled two-qutrit state is less than that for a separable pure two-qutrit state.

As an example, let us consider the following specific families of pure symmetric qutrit states of the form (\ref{symm_qutrits}).

\emph{Example 1.}
Consider the family  of pure states $|\psi_{3\times 3}^{(1)}(t)\rangle$ of the form (\ref{symm_qutrits}) with coefficients 
\begin{equation}
\label{ex1}
    \alpha_{11}(t)=\frac{1-t}{\sqrt{1-2t+2t^2}},\ \ \alpha_{22}(t)=0,\ \ \ \alpha_{33}(t)=\frac{t}{\sqrt{1-2t+2t^2}}\ ,
\end{equation}
where parameter $t\in \left[0,1\right]$. 
By (\ref{ex1}) and (\ref{conc_sym}) the concurrence of $|\psi_{3\times 3}^{(1)}(t)\rangle$ is given by
\begin{equation}
    \mathrm{C}(|\psi_{3\times 3}^{(1)}(t)\rangle)=\frac{2t(1-t)}{1-2t(1-t)}\ .
\end{equation}
Therefore, the pure state $|\psi_{3\times 3}^{(1)}(t)\rangle$ is separable for $t=0,1$, nonseparable for all $t\in\left(0,1\right)$,  its concurrence reaches its (local) maximum at $t=1/2$, as depicted in Fig. 1.

The spin-$1$ CHSH parameter (\ref{1_11}) of the  pure state (\ref{ex1}) can be found by Eq. (\ref{gamma_sym}) and is equal, since  $\alpha_{22}(t)=0$,  to $\gamma_{s=1}(|\psi_{3 \times 3}^{(1)}(t)\rangle)= |\alpha_{11}(t)|^2+|\alpha_{33}(t)|^2$, therefore,
\begin{equation}
    \gamma_{s=1}(|\psi_{3 \times 3}^{(1)}(t)\rangle)=\left(\frac{1-t}{\sqrt{1-2t+2t^2}}\right)^2+\left(\frac{t}{\sqrt{1-2t+2t^2}}\right)^2=1\ ,\   \forall \ t\in\left[0,1\right]\ .
\end{equation}

\begin{figure}[h!]
    \centering
    \includegraphics[width=0.5\textwidth]{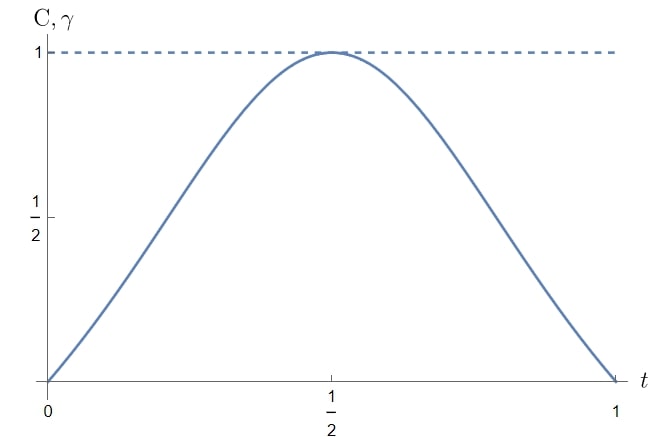}
    \caption{The concurrence (solid) and the spin-$1$ CHSH parameter (dashed) of the pure two-qutrit state (\ref{ex1}) for $t\in\left[0,1\right]$.}
    \label{fig1}
\end{figure}

Here, we have an example of a two-qutrit state with a variable entanglement depending on a parameter $ t\in\left[0,1\right]$ and a constant value of the spin-$1$ CHSH parameter. This phenomenon had already occurred for the Horodecki mixed states (Section \ref{mixed_twoqutrit_states}), but this example shows that it may also occur in case of pure states.

\emph{Example 2.}
Consider the family of pure states $|\psi_{3\times 3}^{(2)}(t)\rangle$ of the form (\ref{symm_qutrits}) with coefficients given by
\begin{equation}
\label{ex2}
    \alpha_{11}(t)=\frac{1-t/2}{\sqrt{1-t+(4/3)t^2}},\ \ \alpha_{22}(t)=\frac{t/2}{\sqrt{1-t+(4/3)t^2}},\ \ \ \alpha_{33}(t)=\frac{t/2}{\sqrt{1-t+(4/3)t^2}}\ ,
\end{equation}
where parameter   $ t\in\left[0,1\right]$. 

By (\ref{ex2}) and (\ref{conc_sym}) the concurrence of this state is equal to
\begin{equation}
    \mathrm{C}(|\psi_{3\times 3}^{(2)}(t)\rangle)=\frac{2t\sqrt{3t^2-8t+8}}{3t^2-4t+4}\ ,
\end{equation}
and is monotonically increasing as shown in Fig. 2. For $t=0$, this state is separable and, for $t=1$, it is maximally entangled.

By (\ref{gamma_sym}), the spin-$1$ CHSH parameter of the  pure symmetric state (\ref{ex2}) is given as  
\begin{equation}
    \gamma_{s=1}(|\psi_{3\times 3}^{(2)}(t)\rangle)=\frac{3}{2} \sqrt{\frac{t^4-4 t^3+9 t^2-8 t+4}{( 4 t^2-3t+3)^2}}\ ,
\end{equation}
for all $t\in \left[0,1\right]$ and is represented in Fig. 2 (dashed curve).
\begin{figure}[h!]
    \centering
    \includegraphics[width=0.5\textwidth]{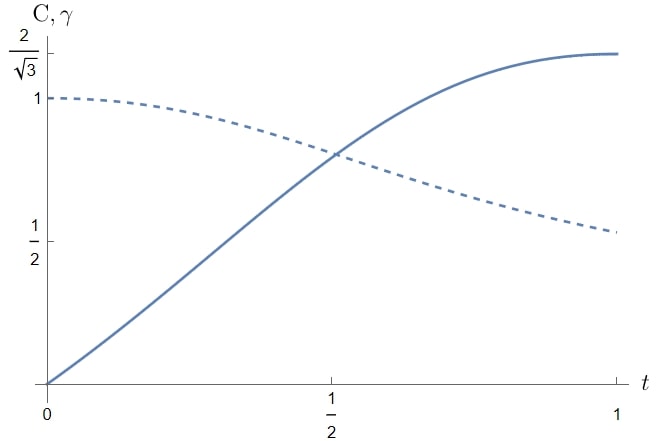}
    \caption{The concurrence  (solid) and spin-$1$ CHSH parameter (dashed) of the pure two-qutrit state (\ref{ex2})   for $t\in\left[0,1\right]$.}
    \label{fig2}
\end{figure}

\noindent In this case, the spin-$1$ CHSH parameter is not constant, it decreases monotonically from $1$ at $t=0$ to $\frac{2\sqrt{2}}{3}$ at $t=1$, as shown in Fig. 2.

In this example, for values of the parameter $t$ in the interval $\left[0,1\right]$, we presented a family of pure states of the form (\ref{symm_qutrits}) for which the entanglement monotonically increases though its spin-$1$ CHSH parameter decreases.

\section{Conjecture}
  
As it is shown analytically in Section \ref{spin_correlation_matrix_section}, under local Alice and Bob spin-$1$ measurements, no one of the nonseparable pure states in families (\ref{anti}) and (\ref{symm_qutrits}) violates the CHSH inequality. The entanglement of these pure states is studied in Propositions \ref{proposition2} and \ref{proposition3} of  Section \ref{entanglement_vs_nonviolation_chsh}. 

Furthermore, within testing of $1,000,000$ randomly generated
nonseparable pure two-qutrit states\footnote{This numerical study has been performed by using Wolfram Mathematica 13.1, see in Appendix B.}, we also have not found a nonseparable pure  two-qutrit state that violates the CHSH inequality under local Alice and Bob spin-$1$ measurements. 

These numerical results and the analytical results in Section \ref{spin_correlation_matrix_section} lead us to
the following conjecture.

\emph{
 Under local Alice and Bob spin-$1$ measurements in an arbitrary nonseparable pure two-qutrit state, the
CHSH inequality is not violated.
}

Note that by Corollary \ref{corollary2} the spin-$1$ CHSH parameter for every separable state is not greater than $1$.
This and the above Conjecture imply that, for every pure two-qutrit state $|\psi_{3\times 3} \rangle\langle \psi_{3\times 3}|$, separable or nonseparable, the spin-$1$ CHSH parameter
\begin{equation}
\gamma_{s=1}(|\psi_{3\times 3}\rangle )\leq 1\ .\label{90}
\end{equation}

Recall that by the spectral theorem every 
 mixed two-qutrit state admits the convex form decomposition 
\begin{equation}
\label{spectral}
    \rho_{3\times 3}^{(mix)}=\sum_k \lambda_k |\phi_{3\times 3}^{(k)} \rangle\langle \phi_{3\times 3}^{(k)}|\ ,\ \lambda_k\geq0\ ,  \sum_k \lambda_k=1,
\end{equation}
where each $\lambda_k$ is an eigenvalue of $\rho_{3\times 3}^{(mix)}$ and $|\phi_{3\times3}^{(k)}\rangle$ is the corresponding eigenvector.

From the convex property (\ref{1_12}) and relations (\ref{90}), (\ref{spectral}) it follows that, for a mixed two-qutrit state, the spin-$1$ CHSH parameter is also not more than one:
\begin{equation}
\label{76}
    \gamma_{s=1}(\rho^{(mix)}_{3\times 3})\leq\sum_{k}\lambda_{k}\text{ }\gamma_{s=1}(|\phi_{3\times 3}^{(k)}\rangle )\leq 1\ .
\end{equation}
 so that by Corollary \ref{corollary1} every mixed two-qutrit state does not violate the CHSH inequality.  
 
Summing up -- based on the above Conjecture, we come to the following statement.

\vspace*{0.1cm}

\emph{Under local Alice and Bob spin-$1$ measurements in any nonseparable two-qutrit state, pure or mixed, the
CHSH inequality is not violated.}

\section{Conclusion}
\label{conclusion}

In the present article, based on the general analytical expression (Proposition \ref{proposition1}) for the maximum  of the CHSH expectation under spin-$1$ measurements, we have analyzed whether or not, under spin-$1$ measurements, the CHSH inequality is violated. 

For a variety of nonseparable two-qutrit states, pure and mixed, we have found analytically (Section 3) 
the values of the spin-$1$ CHSH parameter (\ref{1_11}) specifying violation or nonviolation of the CHSH inequality under local Alice and Bob spin-$1$ measurements. By complementing these analytical results with the numerical study (Appendix B) on the values of this parameter for $1,000,000$ randomly generated pure nonseparable two-qutrit states and taking also into account the spectral decomposition of each mixed state, we put forward the Conjecture (Section 5) that, under local Alice and Bob spin-$1$ measurements in any  nonseparable two-qutrit state, pure and mixed, the CHSH inequality is not violated.

Furthermore, we have also derived in Propositions \ref{proposition2} and \ref{proposition3} (Section 4) the explicit expressions for the values of the concurrence for  pure two-qutrit states in families (\ref{anti}) and \eqref{symm_qutrits} and compared them with the values of the spin-$1$ CHSH parameter for these states. We have found that, in
contrast to spin-$\frac{1}{2}$ measurements, where the spin-$\frac{1}{2}$ CHSH parameter (\ref{0_5}) of a pure
two-qubit state is monotonically increasing with a growth of its concurrence, for a
pure two-qutrit state, this is not the case. In particular, for the two-qutrit GHZ state, which is maximally entangled, the spin-$1$ CHSH parameter (\ref{1_11}) is equal to $\sqrt{\frac{8}{9}}$, while, for some separable pure two-qutrit states, this parameter can be equal to unity. Also,  for each of the Horodecki two-qutrit states \eqref{hor}, the spin-$1$ CHSH parameter is equal by \eqref{gamma_hor} to $4\sqrt{2}/21 < 1$ regardless of the entanglement type of this mixed state according to the classification in \cite{Hor.Hor.Hor:99}.

\section*{Acknowledgment}

The study by E. R. Loubenets in Sections \ref{introduction} and \ref{preliminaries} of this work is supported by
the Russian Science Foundation under Grant No 24-11-00145 and
performed at the Steklov Mathematical Institute of the Russian Academy of
Sciences. The study by E. R. Loubenets and L. Hanotel in Sections \ref{spin_correlation_matrix_section}--\ref{conclusion} was conducted at the HSE University.

\section*{Appendix A}
By taking the partial trace of $|\psi_{3\times3}^{(asym)}\rangle\langle\psi_{3\times3}^{(asym)}|$ over the first space in
$\mathbb{C}^{3}\otimes\mathbb{C}^{3}$, we find the following matrix
representation of the reduced states $\rho_{j}$ (that coincide for $j=1,2$) of $|\psi_{3\times3}^{(asym)}\rangle\langle\psi_{3\times
3}^{(asym)}|$ in the
computational basis in $\mathbb{C}^{3}$.

\begin{equation}
\rho_{j}  = \frac{1}{2}
\begin{pmatrix}
|\alpha_{12}|^{2}+|\alpha_{13}|^{2} & \alpha_{13}\alpha_{23}^{*} &
-\alpha_{12}\alpha_{23}^{*}\\
\alpha_{13}^{*}\alpha_{23} & |\alpha_{12}|^{2}+|\alpha_{23}|^{2} & \alpha
_{12}\alpha_{13}^{*}\\
-\alpha_{12}^{*}\alpha_{23} & \alpha_{12}^{*}\alpha_{13} & |\alpha_{13}%
|^{2}+|\alpha_{23}|^{2}\\
&  &
\end{pmatrix}
\ .
\end{equation}

\section*{Appendix B}
In this Appendix, we present the Mathematica 13.1 code for the numerical study discussed in Section 5. The program consists of the following steps: (i) generating randomly a  unit two-qutrit vector $|\psi_{3\times 3}\rangle \in\mathbb{C}^3\otimes \mathbb{C}^3$; (ii) computing for this state  the spin-$1$ correlation matrix (\ref{1_8}) and its singular values; (iii)  calculating via (\ref{1_11}) the spin-$1$ CHSH parameter $\gamma_{s=1}(|\psi_{3\times 3}\rangle )$ of this state and its concurrence. 

Within $1,000,000$ numerical trials, we have not experienced a case where the parameter $\gamma_{s=1}(|\psi_{3\times 3}\rangle )>1$.

\begin{lstlisting}[language=Mathematica,caption={Numerical study of Section 5}]
Computation of the spin CHSH parameter for random pure two-qutrit states
(*General definitions*)
ccS=Complex[a_,b_]:>Complex[a,-b]; (*Complex conjugate substitution*)
(*Spin-1 Operators*)
s[1]={{0,1,0},{1,0,1},{0,1,0}}/Sqrt[2]; 
s[2]=-I{{0,1,0},{-1,0,1},{0,-1,0}}/Sqrt[2];
s[3]={{1,0,0},{0,0,0},{0,0,-1}};
(*Other useful functions*)
nC[\[Psi]L_]:=\[Psi]L/Sqrt[\[Psi]L.(\[Psi]L/.ccS)] (*This function normalizes every vector \[Psi]L*)
vectorToDensityMatrix[\[Psi]L_]:=KroneckerProduct[\[Psi]L,(\[Psi]L/.ccS)] (*It finds the density operator for the pure state vector \[Psi]L*)
zMatrix[\[Rho]_]:=Table[Tr[\[Rho].KroneckerProduct[s[i],s[j]]],{i,1,3},{j,1,3}] (*It computes the spin-1 correlation matrix of a state \[Rho] *)
The following is the main function, which: (i) takes a randomly generated pure two-qutrit quantum state "\[Psi]L" by randomly generating independent complex numbers for its entries; (ii) normalizes this quantum state; (iii) computes the spin correlation matrix "zM" (17); (iv) computes its singular values "eigvL"  (which are automatically sorted in decreasing order in the case of numerical data) and then the spin CHSH parameter "\[Gamma]" is computed.
main[]:=Module[
{\[Psi]L={RandomComplex[],RandomComplex[],RandomComplex[],RandomComplex[],RandomComplex[],RandomComplex[],RandomComplex[],RandomComplex[],RandomComplex[]},zM,eigvL,\[Gamma]},
zM=zMatrix[vectorToDensityMatrix[nC[\[Psi]L]]]//FullSimplify;
eigvL=Transpose[zM].zM//FullSimplify//Eigenvalues//FullSimplify;
\[Gamma]=(Sqrt[eigvL[[1]]+eigvL[[2]]])
]
This function after 1,000,000 iterations does not find any violation of the CHSH inequality under the conditions described in the article. 	
q=1000000; (*number of iterations*)
eL=ConstantArray[0,q]; (*empty list to store the results of the iterations of the main function*)
For[j=1,j<q+1,j++,eL[[j]]=main[]] (*iteration*)
Max[eL] (*Maximal value of the CHSH parameter for q iterations*) 
0.993671
Computation of concurrence (to generate data for histogram in Fig. 3)
stateToDensityM[\[Psi]_] := KroneckerProduct[\[Psi], \[Psi] /. ccS]
v1=KroneckerProduct[IdentityMatrix[3],{1,0,0}];
v2=KroneckerProduct[IdentityMatrix[3],{0,1,0}];
v3=KroneckerProduct[IdentityMatrix[3],{0,0,1}];
Computation of the reduced density matrix
red[\[Rho]_] := v1.\[Rho].Transpose[v1]+v2.\[Rho].Transpose[v2]+v3.\[Rho].Transpose[v3]
Concurrence
c[\[Psi]_]:=Sqrt[2(1-Tr[red[stateToDensityM[\[Psi]]].red[stateToDensityM[\[Psi]]]])]
eL2=ConstantArray[0,q];  (*empty list*)
For[i= 1,i<q+1,i++,eL2[[i]]=Sqrt[2-2Tr[red[stateToDensityM[eL[[i]][[1]]]].red[stateToDensityM[eL[[i]][[1]]]]]]] (*data of histogram in Fig. 3*)
  \end{lstlisting}

  \vspace*{0.5cm}

The above states are all entangled and in Fig. 3 a histogram of the number of states for a given interval of values of concurrence is presented. 

\begin{figure}[h!]
    \centering
    \includegraphics[width=0.6\textwidth]{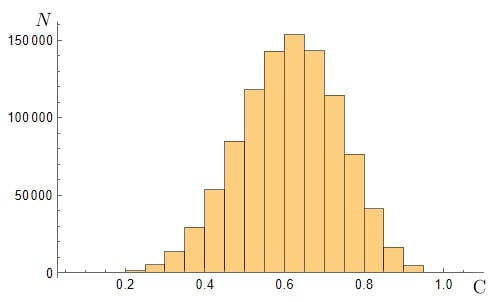}
    \caption{Number of states $N$ among the $1,000,000$ states considered in the sample in each interval of values of the concurrence $\mathrm{C}$.}
    \label{fig3}
\end{figure}

A sample of fifty of these numerical results is shown in Table 1, indicating a randomly generated pure two-qutrit state and its corresponding value of the spin-$1$ CHSH parameter --  according to the results of the program presented above. A randomly generated pure two-qutrit state admits the decomposition  
\begin{equation}
    |\psi_{3 \times 3 }\rangle=\sum_{i,j=1}^3 \psi_{ij}|i\rangle\otimes |j\rangle, \ \psi_{ij}\in\mathbb{C}\ ,
\label{r}
\end{equation}
  and is specified below via the list of its coefficients in (\ref{r}):
  \begin{equation}
      \left\lbrace 
      \psi_{11},\psi_{12},\psi_{13},
      \psi_{21},\psi_{22},\psi_{23},
      \psi_{31},\psi_{32},\psi_{33}
      \right\rbrace\ ,
  \end{equation}
which satisfy the normalization condition $\sum_{i,j=1}^3 |\psi_{ij} |^2=1\ .$

\begin{table}[h!]
{\tiny
\begin{tabular}{l|c}
{\footnotesize $|\psi_{3\times 3}\rangle$} & {\footnotesize$\gamma_{s=1}(|\psi_{3\times 3}\rangle)$ }\\
\midrule
 \{0.23\, +0.02 i,0.21\, +0.26 i,0.26\, +0.28 i,0.38\, +0.34 i,0.16\, +0.12 i,0.02\, +0.22 i,0.05\, +0.39 i,0.22\, +0.35 i,0.\, +0.11 i\} & 0.83 \\
 \{0.29\, +0.26 i,0.33\, +0.25 i,0.22\, +0.1 i,0.03\, +0.32 i,0.31\, +0.26 i,0.32\, +0.22 i,0.19\, +0.3 i,0.04\, +0.17 i,0.09\, +0.22 i\} & 0.91 \\
 \{0.08\, +0.01 i,0.1\, +0.33 i,0.03\, +0.17 i,0.03\, +0.35 i,0.32\, +0.33 i,0.19\, +0.11 i,0.28\, +0.22 i,0.32\, +0.26 i,0.11\, +0.39 i\} & 0.86 \\
 \{0.18\, +0.13 i,0.09\, +0.12 i,0.3\, +0.43 i,0.08\, +0.33 i,0.35\, +0.07 i,0.21\, +0.17 i,0.22\, +0.16 i,0.3\, +0.25 i,0.29\, +0.14 i\} & 0.79 \\
 \{0.22\, +0.28 i,0.33\, +0.23 i,0.\, +0.09 i,0.04\, +0.27 i,0.23\, +0.26 i,0.42\, +0.29 i,0.1\, +0.35 i,0.12\, +0.12 i,0.26\, +0.15 i\} & 0.82 \\
 \{0.31\, +0.24 i,0.15\, +0.24 i,0.16\, +0.07 i,0.32\, +0.16 i,0.06\, +0.35 i,0.01\, +0.36 i,0.2\, +0.11 i,0.4\, +0.15 i,0.09\, +0.32 i\} & 0.83 \\
 \{0.28\, +0.05 i,0.16\, +0.09 i,0.04\, +0.36 i,0.01\, +0.32 i,0.33\, +0.01 i,0.2\, +0.3 i,0.03\, +0.38 i,0.12\, +0.36 i,0.2\, +0.28 i\} & 0.71 \\
 \{0.39\, +0.07 i,0.25\, +0.1 i,0.25\, +0.16 i,0.28\, +0.28 i,0.36\, +0.21 i,0.04\, +0.18 i,0.33\, +0.15 i,0.34\, +0.14 i,0.09\, +0.2 i\} & 0.86 \\
 \{0.02\, +0.34 i,0.2\, +0.12 i,0.16\, +0.2 i,0.\, +0.14 i,0.39\, +0.19 i,0.33\, +0.25 i,0.09\, +0.26 i,0.21\, +0.24 i,0.31\, +0.33 i\} & 0.83 \\
 \{0.1\, +0.07 i,0.08\, +0.38 i,0.04\, +0.15 i,0.4\, +0.14 i,0.19\, +0.34 i,0.22\, +0.07 i,0.1\, +0.25 i,0.37\, +0.35 i,0.16\, +0.25 i\} & 0.81 \\
 \{0.21\, +0.3 i,0.25\, +0.11 i,0.22,0.24\, +0.19 i,0.1\, +0.34 i,0.27\, +0.35 i,0.07\, +0.19 i,0.19\, +0.01 i,0.39\, +0.32 i\} & 0.75 \\
 \{0.37\, +0.24 i,0.03\, +0.01 i,0.1\, +0.25 i,0.39\, +0.4 i,0.42\, +0.02 i,0.06\, +0.07 i,0.05\, +0.19 i,0.15\, +0.19 i,0.1\, +0.36 i\} & 0.69 \\
 \{0.23\, +0.01 i,0.37\, +0.14 i,0.22\, +0.29 i,0.33\, +0.27 i,0.15\, +0.24 i,0.15\, +0.33 i,0.09\, +0.02 i,0.3\, +0.32 i,0.25\} & 0.86 \\
 \{0.32\, +0.38 i,0.25\, +0.25 i,0.09\, +0.29 i,0.02\, +0.22 i,0.31\, +0.07 i,0.\, +0.2 i,0.04\, +0.19 i,0.29\, +0.21 i,0.4\, +0.15 i\} & 0.67 \\
 \{0.06\, +0.36 i,0.38\, +0.14 i,0.32\, +0.16 i,0.07\, +0.38 i,0.14,0.09\, +0.22 i,0.04\, +0.43 i,0.15\, +0.01 i,0.29\, +0.22 i\} & 0.49 \\
 \{0.34\, +0.11 i,0.25\, +0.09 i,0.39\, +0.4 i,0.05\, +0.11 i,0.3\, +0.08 i,0.2\, +0.06 i,0.19\, +0.27 i,0.04\, +0.26 i,0.34\, +0.21 i\} & 0.66 \\
 \{0.38\, +0.33 i,0.1\, +0.17 i,0.22\, +0.25 i,0.38\, +0.25 i,0.14\, +0.06 i,0.\, +0.06 i,0.37\, +0.03 i,0.12\, +0.17 i,0.33\, +0.26 i\} & 0.7 \\
 \{0.06\, +0.24 i,0.24\, +0.19 i,0.08\, +0.39 i,0.02\, +0.15 i,0.41\, +0.15 i,0.38\, +0.22 i,0.25\, +0.05 i,0.27\, +0.07 i,0.32\, +0.19 i\} & 0.85 \\
 \{0.15\, +0.28 i,0.08\, +0.03 i,0.26\, +0.26 i,0.11\, +0.35 i,0.37\, +0.35 i,0.02\, +0.05 i,0.46\, +0.12 i,0.06\, +0.03 i,0.24\, +0.29 i\} & 0.85 \\
 \{0.2\, +0.16 i,0.37\, +0.35 i,0.27\, +0.14 i,0.05\, +0.38 i,0.19\, +0.11 i,0.1\, +0.07 i,0.05\, +0.29 i,0.19\, +0.4 i,0.23\, +0.19 i\} & 0.74 \\
 \{0.4\, +0.27 i,0.29\, +0.01 i,0.11\, +0.16 i,0.27\, +0.06 i,0.36\, +0.22 i,0.26\, +0.05 i,0.08\, +0.26 i,0.26\, +0.13 i,0.22\, +0.34 i\} & 0.89 \\
 \{0.21\, +0.35 i,0.4\, +0.08 i,0.2\, +0.28 i,0.08\, +0.22 i,0.44\, +0.27 i,0.08\, +0.03 i,0.17\, +0.27 i,0.31\, +0.09 i,0.07\, +0.03 i\} & 0.79 \\
 \{0.34\, +0.26 i,0.21\, +0.33 i,0.21\, +0.09 i,0.2\, +0.35 i,0.13\, +0.24 i,0.36\, +0.11 i,0.28\, +0.06 i,0.22\, +0.09 i,0.23\, +0.22 i\} & 0.77 \\
 \{0.21\, +0.23 i,0.25\, +0.36 i,0.16\, +0.26 i,0.18\, +0.05 i,0.02\, +0.03 i,0.32\, +0.24 i,0.27\, +0.22 i,0.11\, +0.31 i,0.38\, +0.2 i\} & 0.52 \\
 \{0.21\, +0.15 i,0.01\, +0.3 i,0.35\, +0.18 i,0.24\, +0.05 i,0.2\, +0.33 i,0.2\, +0.22 i,0.27\, +0.23 i,0.32\, +0.29 i,0.11\, +0.25 i\} & 0.83 \\
 \{0.11\, +0.35 i,0.21\, +0.3 i,0.14\, +0.35 i,0.05\, +0.14 i,0.01\, +0.25 i,0.19\, +0.13 i,0.15\, +0.33 i,0.27\, +0.29 i,0.24\, +0.34 i\} & 0.64 \\
 \{0.17\, +0.11 i,0.22\, +0.04 i,0.1\, +0.19 i,0.08\, +0.38 i,0.33\, +0.2 i,0.37\, +0.26 i,0.39\, +0.1 i,0.07\, +0.25 i,0.35\, +0.07 i\} & 0.77 \\
 \{0.22\, +0.03 i,0.03\, +0.34 i,0.\, +0.04 i,0.01\, +0.01 i,0.14\, +0.35 i,0.04\, +0.33 i,0.27\, +0.35 i,0.35\, +0.27 i,0.3\, +0.32 i\} & 0.84 \\
 \{0.03\, +0.24 i,0.26\, +0.31 i,0.29\, +0.4 i,0.32\, +0.36 i,0.04\, +0.06 i,0.15\, +0.29 i,0.17,0.23\, +0.01 i,0.24\, +0.24 i\} & 0.65 \\
 \{0.11\, +0.37 i,0.21\, +0.1 i,0.1\, +0.37 i,0.14\, +0.06 i,0.22\, +0.33 i,0.09\, +0.12 i,0.35\, +0.02 i,0.31\, +0.27 i,0.37\, +0.15 i\} & 0.79 \\
 \{0.03\, +0.33 i,0.35\, +0.13 i,0.11\, +0.15 i,0.36\, +0.3 i,0.06\, +0.09 i,0.\, +0.36 i,0.09\, +0.33 i,0.16\, +0.03 i,0.3\, +0.35 i\} & 0.49 \\
 \{0.36\, +0.17 i,0.34\, +0.21 i,0.2\, +0.07 i,0.16\, +0.41 i,0.05\, +0.29 i,0.27\, +0.05 i,0.34\, +0.02 i,0.03\, +0.14 i,0.22\, +0.3 i\} & 0.65 \\
 \{0.16\, +0.18 i,0.34\, +0.15 i,0.06\, +0.27 i,0.22\, +0.37 i,0.08\, +0.22 i,0.26\, +0.24 i,0.38\, +0.17 i,0.4\, +0.13 i,0.08\, +0.04 i\} & 0.79 \\
 \{0.3\, +0.39 i,0.14\, +0.05 i,0.18\, +0.17 i,0.12\, +0.06 i,0.21\, +0.08 i,0.39\, +0.07 i,0.02\, +0.37 i,0.38\, +0.39 i,0.09\, +0.1 i\} & 0.64 \\
 \{0.36\, +0.14 i,0.19\, +0.09 i,0.06\, +0.16 i,0.31\, +0.26 i,0.21\, +0.14 i,0.2\, +0.31 i,0.06\, +0.4 i,0.04\, +0.36 i,0.24\, +0.25 i\} & 0.81 \\
 \{0.33\, +0.05 i,0.37\, +0.31 i,0.1\, +0.07 i,0.22\, +0.21 i,0.09\, +0.29 i,0.17\, +0.25 i,0.28\, +0.09 i,0.34\, +0.29 i,0.27\, +0.11 i\} & 0.77 \\
 \{0.04\, +0.03 i,0.02\, +0.38 i,0.32\, +0.37 i,0.28\, +0.31 i,0.02\, +0.04 i,0.21\, +0.13 i,0.28\, +0.28 i,0.22\, +0.28 i,0.12\, +0.29 i\} & 0.61 \\
 \{0.3\, +0.28 i,0.13\, +0.29 i,0.29\, +0.09 i,0.12\, +0.23 i,0.3\, +0.35 i,0.13\, +0.12 i,0.31\, +0.17 i,0.14\, +0.21 i,0.37\, +0.01 i\} & 0.85 \\
 \{0.36\, +0.05 i,0.26\, +0.11 i,0.11\, +0.28 i,0.06\, +0.22 i,0.19\, +0.31 i,0.3\, +0.1 i,0.17\, +0.37 i,0.06\, +0.27 i,0.31\, +0.28 i\} & 0.78 \\
 \{0.23\, +0.31 i,0.25\, +0.11 i,0.16\, +0.29 i,0.31\, +0.36 i,0.19\, +0.14 i,0.07\, +0.08 i,0.26\, +0.16 i,0.27\, +0.26 i,0.34\, +0.12 i\} & 0.78 \\
 \{0.21\, +0.08 i,0.31\, +0.28 i,0.17\, +0.06 i,0.27\, +0.21 i,0.38\, +0.1 i,0.02\, +0.11 i,0.35\, +0.17 i,0.27\, +0.29 i,0.15\, +0.34 i\} & 0.86 \\
 \{0.2\, +0.24 i,0.04,0.09\, +0.06 i,0.38\, +0.41 i,0.1\, +0.14 i,0.13\, +0.44 i,0.07\, +0.24 i,0.05\, +0.43 i,0.19\, +0.23 i\} & 0.73 \\
 \{0.35\, +0.01 i,0.4\, +0.18 i,0.2\, +0.01 i,0.15\, +0.25 i,0.26\, +0.31 i,0.42\, +0.13 i,0.31\, +0.13 i,0.25\, +0.04 i,0.06\, +0.14 i\} & 0.81 \\
 \{0.03\, +0.15 i,0.01\, +0.41 i,0.3\, +0.22 i,0.31\, +0.37 i,0.15\, +0.29 i,0.11\, +0.06 i,0.29\, +0.06 i,0.2\, +0.38 i,0.04\, +0.22 i\} & 0.84 \\
 \{0.14\, +0.16 i,0.21\, +0.38 i,0.01\, +0.39 i,0.08\, +0.22 i,0.32\, +0.32 i,0.18\, +0.01 i,0.19\, +0.27 i,0.13\, +0.38 i,0.22\, +0.06 i\} & 0.79 \\
 \{0.29\, +0.24 i,0.24\, +0.31 i,0.13\, +0.38 i,0.18\, +0.07 i,0.32\, +0.25 i,0.01\, +0.23 i,0.12\, +0.14 i,0.24\, +0.32 i,0.3\, +0.04 i\} & 0.84 \\
 \{0.19\, +0.32 i,0.2\, +0.38 i,0.33\, +0.15 i,0.2\, +0.23 i,0.31\, +0.05 i,0.38\, +0.21 i,0.13\, +0.17 i,0.18\, +0.11 i,0.08\, +0.27 i\} & 0.82 \\
 \{0.36\, +0.19 i,0.31\, +0.01 i,0.19\, +0.19 i,0.35\, +0.34 i,0.06\, +0.21 i,0.13\, +0.38 i,0.21\, +0.06 i,0.1\, +0.04 i,0.18\, +0.35 i\} & 0.57 \\
 \{0.28\, +0.26 i,0.04\, +0.19 i,0.03\, +0.31 i,0.29\, +0.19 i,0.22\, +0.32 i,0.16\, +0.36 i,0.08\, +0.3 i,0.14\, +0.2 i,0.36\, +0.07 i\} & 0.84 \\
 \{0.26\, +0.4 i,0.18\, +0.37 i,0.1\, +0.06 i,0.38\, +0.33 i,0.06\, +0.14 i,0.02\, +0.17 i,0.18\, +0.33 i,0.01\, +0.04 i,0.28\, +0.24 i\} & 0.79 \\
\end{tabular}}
\caption{Numerical examples of random pure two-qutrit states with complex coefficients (left column) and the values (right column) of the spin-$1$ CHSH parameter for these states. Due to space limitations, we present here all numerical results up to two decimal digits.}
\end{table}


\begin{thebibliography}{99}    
 \bibitem {Lou:08}Loubenets E R 2008 \textit{J. Phys. A: Math. and
Theor.} \textbf{41} 445304

\bibitem {Cla.Hor.Shi:69}Clauser J F, Horne M A, Shimony A and Holt R A 1969
\emph{Phys. Rev. Lett.} \textbf{23} 8804

\bibitem{Loub2020} Loubenets E R 2020 \emph{
J. Phys. A: Math. Theor.} \textbf{53} 045303

\bibitem {7}Tsirelson B S 1980 \emph{Lett. Math. Phys.} \textbf{4} 93

\bibitem {8}Tsirelson B S 1987 \emph{J. Soviet Math.} \textbf{36} 557-570

\bibitem {singlet}Peruzzo G, Sorella S P 2023 \emph{Physics Letters A} \textbf{474} 128847

\bibitem {Lou.Kuz.Han:24}Loubenets E R, Kuznetsov S and Hanotel L 2024 \emph{
J. Phys. A: Math. Theor.} \textbf{57} 055302

\bibitem{Hil.Woo:97} Hill S A and Wootters W K 1997 \emph{Phys. Rev. Lett.} \textbf{78} 5022

\bibitem{Run.Buz.etal:01} Rungta P, Buzek V, Caves C M, Hillery M and Milburn G J 2001
\emph{Phys. Rev. A} \textbf{64} 042315

\bibitem {Che.Alb.Fei:05}Chen K, Albeverio S, and Fei S M 2005 \emph{Phys. Rev. Lett.} \textbf{95} 040504

\bibitem {Ver.Wol:02}Verstraete F and Wolf M M 2002 \emph{Phys. Rev. Lett.} \textbf{89} 170401

\bibitem{1} Gisin N and Peres A 1992 \emph{Phys. Lett. A} \textbf{162} 15--17

\bibitem {arXiv} Loubenets E R and Hanotel L 2024 arXiv preprint arXiv:2412.03470

\bibitem{Hor.Hor.Hor:95} Horodecki R, Horodecki P and Horodecki M 1995
\emph{Phys. Lett. A } \textbf{200} 340--344

\bibitem{qkd} Pastorello D 2017 \emph{Int. J. Quantum Inf.} \textbf{15} 1750040

\bibitem{ferm} Kuno Y 2024 \emph{J. Phys.: Condens. Matter} \textbf{36} 505401

\bibitem{sq} Ait Chlih A, Habiballah N and Khatib D 2024 \emph{Int. J. Mod. Phys. B} \textbf{38} 2450310

\bibitem{LHC} Fabbrichesi M, Floreanini R and Panizzo G 2021 \emph{Phys. Rev. Lett.} \textbf{127} 161801

\bibitem{rev} Barr A J, Fabbrichesi M, Floreanini R, Gabrielli E, and Marzola L 2024 \emph{Prog. Part. Nucl. Phys.} \textbf{139} 104134

\bibitem{higgs} Aguilar-Saavedra J A, Bernal A, Casas J A and Moreno J M 2023 \emph{Phys. Rev. D} \textbf{107} 016012

\bibitem{BellLHC} Fabbrichesi M, Floreanini R, Gabrielli E and Marzola L 2023 \emph{Eur. Phys. J. C} \textbf{83} 823

\bibitem{vbs} Morales R A 2023 \emph{Eur. Phys. J. Plus} \textbf{138} 1-24

\bibitem{charm} Fabbrichesi M, Floreanini R, Gabrielli, E and Marzola L 2024 \emph{Phys. Rev. D} \textbf{110} 053008

\bibitem{boj} Fabbrichesi M, Floreanini R, Gabrielli, E and Marzola L 2024 \emph{Phys. Rev. D} \textbf{109} L031104

\bibitem{bff} Gabrielli, E and Marzola L 2024 \emph{Symmetry} \textbf{16} 1036


\bibitem{Wer:89} Werner R F 1989 \emph{Phys. Rev. A} \textbf{40} 4277

\bibitem{Hor.Hor.Hor:99}Horodecki P, Horodecki M and Horodecki R 1999
\emph{Phys. Rev. Lett.} \textbf{82} 1056

\bibitem{Lou:05} Loubenets, E R 2005
\emph{J. Phys. A: Math. Gen.} \textbf{38} L653

\end{thebibliography}
\end{document}